\documentclass[numsec,webpdf,traditional,mediumone]{oup-authoring-template}

\graphicspath{{Fig/}}

\usepackage{amsmath}
\usepackage{amsfonts}
\usepackage{amssymb}
\usepackage{graphicx}
\usepackage{float}
\usepackage{booktabs,tabularx}
\usepackage{multirow}
\usepackage{siunitx} 
\usepackage{rotating}  
\usepackage{pdflscape} 

\makeatletter

\@ifundefined{description}{%
  }{}
\makeatother
\usepackage{enumitem}
\usepackage{tcolorbox}
\usepackage{algorithm}
\usepackage{algpseudocode} 
\usepackage{tikz}
\usetikzlibrary{arrows.meta, positioning}

\tikzset{
  var/.style   = {circle, draw, minimum size=12mm, inner sep=0pt},
  fixed/.style = {rectangle, draw, minimum height=8mm, minimum width=8mm, inner sep=2pt},
  >=Stealth
}

\newlist{assump}{enumerate}{1}
\setlist[assump]{label=\textbf{A\arabic*.},ref=A\arabic*,
                 leftmargin=*,      
                 labelsep=0.6em,    
                 itemsep=0.25ex,    
                 topsep=0.5ex,      
                 align=left}        

\def\bSig\mathbf{\Sigma}

\newcommand{\abs}[1]{\lvert #1 \rvert}

\newcommand{\R}{\mathbb{R}}

\algrenewcommand\algorithmicrequire{\textbf{Require:}}
\algrenewcommand\algorithmicensure{\textbf{Ensure:}}
\algnewcommand\algorithmicinitialize{\textbf{Initialize:}}
\algrenewcommand\algorithmiccomment[1]{\hfill\(\triangleright\) #1}
\algnewcommand\Initialize{\item[\algorithmicinitialize]}

\theoremstyle{thmstyleone}%
\newtheorem{theorem}{Theorem}
\newtheorem{lemma}[theorem]{Lemma}%
\theoremstyle{thmstyletwo}%
\theoremstyle{thmstylethree}%

\begin{document}

\journaltitle{Biostatistics}
\DOI{DOI added during production}
\copyrightyear{2026}
\pubyear{YEAR}
\vol{XX}
\issue{x}
\access{Published: Date added during production}
\appnotes{Paper}

\firstpage{1}


\title[Mediation Analysis for ITRs]{Bayesian Mediation Analysis for Individualized Treatment Rules}

\author[1,$\ast$]{Emmanuel M. Rockwell}
\author[2]{Patrick J. Smith}
\author[1]{Michael R. Kosorok}
\author[3]{Nikki L. B. Freeman}

\address[1]{\orgdiv{Department of Biostatistics}, \orgname{University of North Carolina at Chapel Hill}, \orgaddress{\street{135 Dauer Drive, 3103 McGavran-Greenberg Hall CB 7420, Chapel Hill}, \postcode{27510}, \state{NC}, \country{United States}}}
\address[2]{\orgdiv{Department of Psychiatry}, \orgname{University of North Carolina at Chapel Hill}, \orgaddress{\street{333 S. Columbia Street, Suite 304 MacNider Hall, Chapel Hill}, \postcode{27514}, \state{NC}, \country{United States}}}
\address[3]{\orgdiv{Department of Biostatistics and Bioinformatics}, \orgname{Duke University}, \orgaddress{\street{2424 Erwin Road, Suite 1102 Duke University Medical Center, Durham}, \postcode{27710}, \state{NC}, \country{United States}}}

\corresp[$\ast$]{Corresponding author. \href{email:erockwell@unc.edu}{erockwell@unc.edu}}

\received{Date}{0}{Year}
\revised{Date}{0}{Year}
\accepted{Date}{0}{Year}



\abstract{The value of an individualized treatment rule (ITR), defined as the expected outcome under treatment assignment according to the rule, is useful for assessing average clinical benefit but does not explain how the benefit of a rule is generated. We propose a causal mediation framework for decomposing the value contrast between a prespecified candidate ITR and a clinically meaningful reference rule into direct and indirect components. Using rule-specific nested potential outcomes, we define natural direct and indirect rule effects that quantify the extent to which the improvement in value arises through pathways operating directly on the outcome versus through a specified mediator. We give identification conditions under which these components are identified by a rule-level mediation g-formula. For estimation, we adapt Bayesian causal mediation forests to obtain posterior inference for the value contrast and its path-specific components. Our simulations demonstrate that the proposed estimator achieved near-nominal credible interval coverage with decreasing bias and root mean squared error as sample size increased in settings with varying direct and mediated contributions. We further illustrate the method using data from the TRIUMPH trial, decomposing the cognitive benefit of a lifestyle intervention rule through candidate neurovascular, cardiorespiratory, and behavioral mediators. The proposed framework complements optimal ITR learning with explanation using mediation, providing a natural approach for mechanistic evaluation of ITRs.} 



\maketitle


\section{Introduction}
\label{s:intro}

A central aim of personalized medicine is the development of individualized treatment rules (ITRs) which assign treatment as a function of baseline patient covariates to maximize a prespecified clinical benefit \citep{Kosorok2015-hf, Mirnezami2012-gj, Chakraborty2013-yq, Li2023-qg}. Throughout, we use the term individualized treatment rule while acknowledging that related fields may use the term treatment policy. Advances in statistical precision medicine have developed useful methods for estimating optimal treatment rules \citep{Murphy2003-sm, Schulte2014-zb, Zhao2012-yr, Qian2011-sr}. While these approaches may yield rules that improve outcomes on average, they often offer limited insight into the underlying mechanisms through which a treatment affects outcomes. Nonetheless, in many scientific applications and clinical contexts, investigators are not only interested in whether a rule performs well, but how benefit is generated through underlying biological or behavioral pathways. In these settings, good decision rules are not sufficient on their own. To justify clinical implementation and generate new scientific understanding, investigators must also understand the mechanisms through which a rule improves outcomes.

The TRIUMPH trial investigated the effect of intensive lifestyle intervention (versus standardized education and physician advice) on resistant hypertension, with secondary outcomes measuring cognitive performance. Investigators found that the intervention resulted in clinically meaningful reductions in blood pressure and improvement in cognitive outcomes, but that treatment response was highly heterogeneous. For example, \citet{Avorgbedor2023-wj} found that baseline systemic inflammation significantly moderated the cognitive benefit of lifestyle intervention, with larger improvements among participants with greater pre-intervention inflammation. This heterogeneity in response motivates the use of ITRs to determine which patients are most likely to benefit from intensive lifestyle intervention. However, estimating an optimal ITR only addresses the question of who to treat with which intervention. Equally important from mechanistic and translational perspectives is why an ITR's recommended treatment improves outcomes among the selected patients. For example, the benefit of lifestyle intervention may operate through changes in blood pressure, weight, vascular function, inflammation, physical activity, or other behavioral and biological pathways, and the relative importance of these pathways may vary across the same patient characteristics used to assign treatment. In this work, we develop a causal mediation framework to evaluate the value of a prespecified individualized treatment rule and to decompose its benefit into direct and indirect components.

Mechanistic evaluation of a treatment rule can also extend beyond the sample from which the rule was learned. While a complete account of the pathways through which an ITR generates benefit is not necessarily a prerequisite for clinical adoption, this understanding can nonetheless aid interpretation and help clinicians and patients reason about why a rule works. Mechanistic insight may also inform transportability. When a rule's benefit operates largely through a particular mediator, its performance in a new population depends in part on that mediator's distribution in that setting. Decomposing the rule's value into pathway-specific components can help distinguish whether benefit is primarily attributable to treatment effects operating through specified mediators, the direct pathway, or both. Such information may refine mechanistic hypotheses, support clinical interpretation, and inform expectations about external validity. 

The remainder of the paper is organized as follows. Section \ref{s:background} reviews relevant background on ITRs, causal mediation analysis, and Bayesian estimation of mediation estimands. In Section \ref{s:methods}, we formalize notation, make foundational causal identification assumptions, define the direct and indirect components of the value function, and introduce the causal model. Section \ref{s:estimation} presents the estimation strategy using g-formula representations. Sections \ref{s:simulation} and \ref{s:application} report the results of the simulation study and a motivating application utilizing data from the TRIUMPH trial, which explores the effect of lifestyle interventions on cognitive function. The trial collects metabolic, vascular, and behavioral data from patients both at baseline and post-randomization generating a natural environment for exploring potential mediators of cognitive benefit. In Section \ref{s:discussion}, we conclude with a discussion of possible extensions, including dynamic treatment regimes and longitudinal mediation.

\section{Background}
\label{s:background}

A common goal of precision medicine is estimating ITRs, which tailor treatments to the unique health status of an individual by regarding patient heterogeneity as an informative feature rather than a nuisance \citep{Kosorok2019-ov}.  It is reasonable to expect that the covariates along which ITRs differentiate patient decisions may reflect variation in the biological or behavioral processes through which treatment affects the outcome. As a result, personalized medicine is closely related to subgroup analysis and heterogeneous treatment effect estimation, which also seeks to characterize variation in treatment response across patients \citep{Wager2015-kw, Hahn2017-dv, Rekkas2020-kq, Caron2022-fc, Han2023-mc, Hernan2016-sm}. However, identifying covariates that predict treatment benefit is distinct from explaining the causal mechanisms through which that benefit arises. Thus, although ITRs assign treatment based on heterogeneity in patient response, additional causal structure is needed to determine whether a rule improves outcomes directly via treatment, through intermediate biological or behavioral processes, or through both.

Mediation analysis studies how the effect of an exposure on an outcome operates both directly and through specified intermediate variables known as mediators. This typically involves decomposing the treatment effect into a direct effect, the component of the treatment effect measuring the influence of the exposure on the outcome holding other factors fixed at the control level (i.e., not operating through mediating variables), and an indirect effect, the component operating through the mediator of interest \citep{MacKinnon2012-oz, VanderWeele2016-ut, Imai2010-nc}. Foundational work in mediation analysis was rooted in linear structural equation models (LSEMs) \citep{Baron1986-si, MacKinnon2007-dd, Lee2019-wb}. More recent developments place mediation analysis within a causal framework \citep{Robins1992-vq, Pearl2001-bo}, formalizing natural direct and indirect effects in terms of causal pathways and specifying identification conditions for estimating these effects \citep{Tchetgen2012-cp, Avin2005-sv, Malinsky2019-po}. This work has expanded to accommodate multiple mediators, complex outcomes, intricate interaction structures, intermediate confounders, and Bayesian estimation strategies \citep{Miocevic2018-ev, Roy2022-vm, Sun2024-nm, Valeri2013-jd, Diaz2021-dk, Chen2024-sh}. 

A treatment effect is defined as a contrast between potential outcomes $Y_i(1)$ and $Y_i(0)$, only one of which is observed for any given unit \citep{Holland1986-ex}. Mediation estimands introduce additional counterfactual structure. Natural direct and indirect effects are defined through nested counterfactual outcomes. Note that the outcome under one treatment level and mediator under another cannot be jointly observed, even under randomization. Their identification therefore requires assumptions linking these counterfactual quantities to the observed data distribution.

Although identification results in mediation analysis are often stated nonparametrically, estimation typically requires modeling several nuisance functions, including the conditional distribution of the mediator and the conditional mean of the outcome given treatment, mediator, and baseline covariates \citep{Imai2010-nc, Tchetgen2012-cp}. Flexible modeling is attractive because it can accommodate nonlinearities and complex interactions within the covariate structure when estimating these nuisance functions. Furthermore, causal estimation requires additional care because regularized prediction methods can introduce bias in causal contrasts through phenomena such as regularization-induced confounding \citep{Hahn2017-dv}. These considerations have motivated Bayesian and machine learning approaches that are useful for flexibly modeling the nuisance functions that arise in mediation \citep{Chipman2010-cx, Hill2020-rq, Hahn2017-dv}. One foundational approach includes using Bayesian additive regression trees (BART), which provide a powerful nonparametric method for regularized conditional mean modeling with uncertainty quantification \citep{Chipman2010-cx, Hill2020-rq}, and have been widely adopted in causal inference for estimating heterogeneous treatment effects \citep{Hahn2017-dv, Wager2015-kw, Dorie2022-lt, Caron2022-fc}. Recent work has adapted these ideas to mediation settings. Specifically, \citet{Linero2025-it} proposed Bayesian causal mediation forests (BCMFs), a BART-based framework for estimating average causal mediation effects.

To date, most mediation methodology has focused on decomposing the average treatment effect of a fixed exposure into direct and indirect components \citep{MacKinnon2012-oz, Ge2023-ac}. In contrast, most work on ITRs has focused on characterizing and optimizing the value of a rule, typically defined as the expected outcome under the given rule. However, to our knowledge, no existing method has explored how to decompose the value of a prespecified treatment rule relative to a clinically meaningful reference rule into direct and indirect mechanistic components. The goal of this research is to quantify how much of the improvement delivered by a rule relative to a reference rule (such as the standard of care) arises from its direct influence on the outcome versus its effect on mediating variables, thereby extending mediation analysis to evaluate individualized treatment rules.

\section{Method Framework}
\label{s:methods}

Our objective is to determine how much of the value contrast $V(d)-V(r)$ is attributable to pathways operating through a mediator given a candidate ITR $d$ and a clinically meaningful reference rule $r$. We formalize this question by defining rule-specific estimands, identification assumptions, and the g-formula representation for the direct and indirect components of this value contrast.

\subsection{Setting and notation}

Because both mediation and ITRs require estimation of counterfactual quantities, we will adopt the potential outcomes framework commonly used in causal inference \citep{Rubin1974-wh, Rubin2004-er, Rubin2005-iy}. Capital letters denote random variables. Variables without a subscript are vectors, and variables with a subscript are scalar-valued. 

Let $X\in\mathcal X$ denote a vector of baseline covariates, $A\in \mathcal{A} = \{0,1\}$ denote a binary treatment, $M\in\mathcal M\subseteq \R$ denote a post-treatment mediating variable, and $Y\in \mathcal Y\subseteq \R$ denote the outcome of interest. For each individual $i$, we write $M_i(a)$ for the potential value of the mediator under treatment $a \in \{0,1\}$, and $Y_i(a,m)$ for the potential value of the outcome under treatment $a$ and mediator level $m$. 

Let $\mathcal D$ denote a class of measurable binary treatment rules and suppose $d: \mathcal X\to \{0,1\} \in\mathcal D$ is an individualized treatment rule of interest mapping from the space of covariates to the space of available treatment recommendations. Further suppose $r: \mathcal X \mapsto \{0,1\}$ is a reference rule. This rule serves as a baseline against which we evaluate $d$. Therefore, some appropriate choices for $r$ include the standard of care, the rule that treats no one, or the rule that treats everyone. 

Given two rules $d,r \in \mathcal D$, we define rule-specific mediator and outcome counterfactuals as $M_i^r(X_i) \equiv M_i\{r(X_i)\}$ and $Y_i^{d,r}(X_i) \equiv Y_i\left\{d(X_i), M_i^r(X_i)\right\}$. The value of an ITR $d$ is defined as the mean outcome under the rule,
\[V(d) \equiv E[Y\{d(X), M(d(X))\}] = E\left[ Y^{d,d}(X) \right],\]
which we interpret as the average clinical benefit obtained by treating the entire population according to rule $d$. Our goal is to decompose the value contrast $V(d) - V(r)$ between a candidate rule $d$ and a reference rule $r$ into components that reflect distinct causal pathways through the mediator. This value contrast encodes the average improvement we would see in the outcome if we treated the entire population under $d$ rather than under $r$. If $r$ is a clinically well-understood rule, then a decomposition of the value contrast will help isolate the underlying causal pathways along which the new rule $d$ is realizing a benefit over $r$.

In this paper, we treat $d$ and $r$ as prespecified rules to be evaluated. The candidate rule $d$ may be derived from clinical knowledge or a policy learning procedure. If $d$ is estimated using the same data used for evaluation, additional sample-splitting or inferential adjustment would be needed to account for estimation uncertainty.

\subsection{Mediation Estimands}

Figure \ref{fig:dag-example} is a directed acyclic graph illustrating a causal mechanism between $X$, $A$, $M$, and $Y$. The direct pathway is colored blue, and the mediated pathway is colored red. 

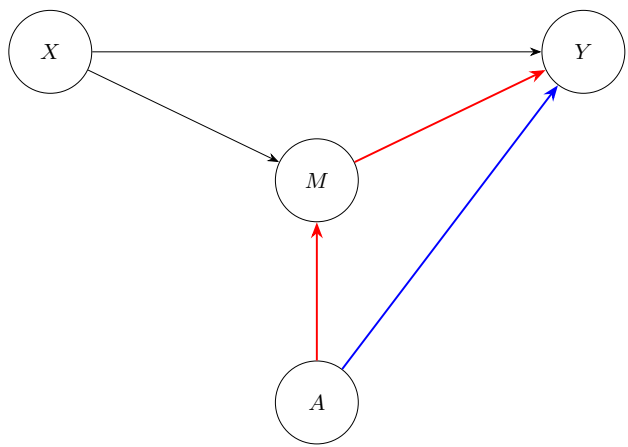
\begin{figure}[htbp]
  \centering
  \resizebox{0.6\linewidth}{!}{%
  \begin{tikzpicture}[node distance=2cm]
    \node[var] (X) {$X$};
    \node[var, below right=1cm and 3cm of X] (M) {$M$};
    \node[var, above right=1cm and 3cm of M] (Y) {$Y$};
    \node[var, below=2cm of M] (A) {$A$};

    \draw[->] (X) -- (M);
    \draw[->] (X) -- (Y);
    \draw[->, red, thick] (A) -- (M);
    \draw[->, blue, thick] (A) -- (Y);
    \draw[->, red, thick] (M) -- (Y);
  \end{tikzpicture}}
  \caption{Directed acyclic graph showing relationships among $X$, $A$, $M$, and $Y$. 
  Red arrows denote the indirect (mediated) path and the blue arrow denotes the direct path.}
  \label{fig:dag-example}
\end{figure}

Classical mediation analysis decomposes the average treatment effect into a natural direct effect and a natural indirect effect. We define the natural direct effect as the effect of changing treatment from 0 to 1 while holding the mediator fixed at the value it would have taken under the control \citep{Robins1992-vq, Pearl2001-bo, Hernan2020-ku}. We take care to distinguish the term \textit{natural} direct effect as opposed to \textit{controlled} direct effect, which is the effect of changing the treatment level while keeping the mediator fixed at the same level for everyone. While controlled direct effects are useful in prescriptive settings (e.g., modifying a relationship to match policy objectives), natural direct effects are better suited for descriptive settings \citep{Pearl2001-bo}. Mathematically, we can write the natural direct effect as
\[\zeta(a) = E\left[Y_i\{1, M_i(a)\} - Y_i\{0, M_i(a)\}\right]\]
We define the natural indirect effect as the effect of changing the mediator from the value it would take under treatment 0 to the value it would take under treatment 1 while holding treatment itself fixed. The natural indirect effect isolates the treatment effect operating through the mediator pathway. Note that ``controlled" indirect effects cannot be defined using the direct formulation, because the quantity $E[Y_i(a, m) - Y_i(a, m')]$ can be non-zero even when the treatment has no effect on the mediator \citep{Pearl2001-bo, Cinelli2025-my}. Therefore, the descriptive setting of ``natural" effects using nested potential outcomes is the preferred formulation of the two that correctly describes the quantity of interest. Mathematically, we can write the natural indirect effect as
\[\eta(a) = E\left[Y_i\{a, M_i(1)\} - Y_i\{a, M_i(0)\}\right].\]
so that the total effect satisfies $E\{Y_i(1,M_i(1)) - Y_i(0,M_i(0))\} = \zeta(0) + \eta(1) = \zeta(1) + \eta(0)$.

Rather than decomposing the average treatment effect, here our goal is to decompose the value contrast of a proposed rule. Let $d(x)\in\{0,1\}$ be the rule of interest and let $r(x)$ be a reference rule. Using the rule-specific potential outcomes introduced above, we define the following four outcomes: 
\begin{itemize}[leftmargin=1cm, labelsep=0.2cm]
    \item $Y^{d,d}(X) = Y\left(d(X), M\{d(X)\}\right)$ is the outcome under rule $d$ such that both the treatment and mediator are assigned according to $A = d(X)$. 
    \item $Y^{r,r}(X) = Y\left(r(X), M\{r(X)\}\right)$ is the outcome under the reference rule along both pathways.
    \item $Y^{d,r}(X) = Y\left(d(X), M\{r(X)\}\right)$ is the cross-world counterfactual outcome where the direct path is assigned $d(X)$ and the indirect path is assigned $r(X)$.
    \item $Y^{r,d}(X) = Y\left(r(X), M\{d(X)\}\right)$ is the cross-world counterfactual outcome where the direct path is assigned $r(X)$ and the indirect path is assigned $d(X)$.
\end{itemize}

We refer to $Y(1,M(0))$ and $Y(0,M(1))$ and their rule-level analogues as \textit{cross-world counterfactuals} because they combine the outcome response under one treatment assignment with the mediator value that would have arisen under another treatment assignment \citep{Hernan2020-ku}. Unlike ordinary potential outcomes such as $Y(1)$ and $Y(0)$, these cross-world counterfactuals cannot be empirically confirmed by an experimental intervention as they cannot occur simultaneously in the same individual. Nonetheless, these quantities are used in natural direct and indirect effect decompositions because they provide a formal way to separate the component of an effect operating through a mediator from the component operating directly on the outcome. This choice targets a natural effect or path-specific effect interpretation \citep{Tchetgen2012-cp, Shpitser2018-pl}. Other mediation estimands, including interventional direct and indirect effects and separable effects, avoid cross-world counterfactuals by targeting different hypothetical interventions or decompositions \citep{Robins2010-bt, Robins2022-mr}. We focus on the natural rule-level decomposition because it directly partitions the rule value contrast into components corresponding to the direct and indirect pathways in Figure \ref{fig:dag-example}. As we will later discuss, they can be estimated using the g-formula by introducing an additional independence assumption for identification. 

These rule-level expected outcomes allow us to decompose the rule contrast $V(d) - V(r)$ into two components by introducing four contrasts:
\begin{align*}
  \Delta_D^{(d)}(d,r) &= E\left\{ Y^{d,d}(X) - Y^{r,d}(X) \right\}, \\
  \Delta_D^{(r)}(d,r) &= E\left\{ Y^{d,r}(X) - Y^{r,r}(X) \right\}, \\
  \Delta_I^{(d)}(d,r) &= E\left\{ Y^{d,d}(X) - Y^{d,r}(X) \right\},\;\;\textrm{and}\\
  \Delta_I^{(r)}(d,r) &= E\left\{ Y^{r,d}(X) - Y^{r,r}(X) \right\}.
\end{align*}

We refer to $\Delta_D^{(d)}(d,r)$ and $\Delta_D^{(r)}(d,r)$ as the \textit{natural direct rule effect}, capturing the increase in value changing from rule $r$ to $d$ only along the pathway from treatment to outcome while holding the treatment along the mediator path constant. Analogously, $\Delta_I^{(d)}(d,r)$ and $\Delta_I^{(r)}(d,r)$ is the \textit{natural indirect rule effect}, capturing the component of the value contrast attributable to the mediated pathway through $M\{d(X)\}$. Then, by construction, there are two possible sequential decompositions of the value contrast given as
\[V(d) - V(r) = \Delta_D^{(r)}(d,r) + \Delta_I^{(d)}(d,r) = \Delta_D^{(d)}(d,r) + \Delta_I^{(r)}(d,r).\]

For the remainder of the paper, I will use the decomposition $V(d) - V(r) = \Delta_D^{(r)}(d,r) + \Delta_I^{(d)}(d,r)$ and simplify notation by writing $\Delta_D^{(r)}(d,r)\equiv \Delta_D(d,r)$ and $\Delta_I^{(d)}(d,r) \equiv \Delta_I(d,r)$.

\subsection{Identification assumptions}
\label{subsec:identification}

The challenge posed by natural mediation estimands is stronger than the usual missing data problem for treatment effects. For an ordinary treatment effect, each individual has potential outcomes $Y(1)$ and $Y(0)$, only one of which is observed. This is the \textit{fundamental problem of causal inference} asserted by Paul Holland \citep{Holland1986-ex}. In contrast, cross-world quantities such as $Y\{1,M(0)\}$ and $Y\{0,M(1)\}$ do not generally correspond to an intervention that could be directly implemented in a randomized trial. Identification of these quantities therefore requires assumptions that connect the nested counterfactuals to the observed data distribution.

Identification of the proposed rule-level mediation estimands relies on standard assumptions from the causal mediation literature \citep{Rubin2005-iy}. We adopt consistency, exchangeability, and positivity conditions as working assumptions \citep{Robins1992-vq, Imai2010-fz}. Under the assumptions below, these rule-level mediation estimands are causally identified by the mediation g-formula:
\begin{enumerate}[label=\textbf{A\arabic*.},leftmargin=1cm, labelsep=0.2cm]
\item (Consistency) When $A_i = a$ is observed, then $M_i = M_i(a)$. Additionally, when $A_i = a$ and $M_i = m$ are observed, then $Y_i = Y_i(a,m)$.

\item (Treatment Exchangeability) For all $a,a' \in \{0,1\}$, $m\in\mathcal M$, and $x \in \mathcal X$,
  \[ \{Y_i(a',m), M_i(a)\} \perp A_i \mid X_i = x.\]
  
\item (Cross-world Exchangeability) For all $a,a' \in \{0,1\}$, $m \in \mathcal M$, and $x \in \mathcal X$,
  \[Y_i(a',m) \perp M_i(a) \mid X_i = x.\]
  
\item (Positivity) For all $x \in \mathcal X$, $0 < P(A_i = a \mid X_i = x) < 1$ such that $a \in \{0,1\}$ and the mediator support under each treatment level overlaps sufficiently to evaluate the cross-world means appearing in the mediation g-formula.
\end{enumerate}

Together assumptions A1-A4 enable identification of the distribution of $\{Y_i(a,M_i(a')): a,a' \in \{0,1\}\}$ via the g-formula. That is, for any pair of rules \(d,r\in\mathcal D\),
\begin{align*}
   E\{Y^{d,r}(X)\}& = \int_{\mathcal X} \int_{\mathcal M} E(Y\mid A=d(x),M=m,X=x)\\
   &\qquad \times f_{M\mid A,X}\{m\mid A=r(x),X=x\} dm dx. 
\end{align*} 

In particular, the marginal distributions of $Y^{d,d}(X)$, $Y^{r,r}(X)$, $Y^{d,r}(X)$ and $Y^{r,d}(X)$ are identified for any pair of rules $(d,r)$ that map to $\{0,1\}$, and hence so are $V(d)$, $V(r)$ and the decomposition $V(d)-V(r)=\Delta_D(d,r)+\Delta_I(d,r)$. The proof of this statement is provided in the appendix. Throughout, we view these assumptions as idealized working conditions. However, in applications, we recommend sensitivity analyses to violations of A2 and A3 \citep{VanderWeele2010-nh, Valeri2013-jd}.

\section{Estimation of Effects}
\label{s:estimation}

The identifying conditions of Section \ref{subsec:identification} show that estimation of the rule-level mediation components requires models for the mediator distribution $f_{M\mid A,X}$ and the outcome regression $E(Y\mid M,A,X)$. We estimate these components using Bayesian causal mediation forests (BCMF) developed by \citet{Linero2025-it}, an extension of Bayesian additive regression trees (BART) \citep{Chipman2010-cx, Hill2020-rq, Hahn2017-dv, Linero2022-ae} adapted to mediation analysis. While other mediation estimation methods are suitable for this task, we choose BCMFs due to their flexibility in modeling nonlinearities and interactions, as well as their ability to quantify uncertainty. 

\subsection{Bayesian tree ensemble estimation}

To estimate the rule-level mediation components $\Delta_D(d,r)$ and $\Delta_I(d,r)$, we extend Bayesian causal mediation forests (BCMFs) \citep{Linero2025-it, Ting2025-xm}, a Bayesian tree ensemble approach that flexibly models the mediator and outcome regressions needed for g-formula estimation. Let
\begin{align*}
  M_i(a) \mid X_i &= m(a,X_i) + \varepsilon_{M_i},\\
  Y_i(a,m) \mid X_i &= y(m,a,X_i) + \varepsilon_{Y_i},
\end{align*}
where $m(\cdot)$ and $y(\cdot)$ are unknown regression functions and $\varepsilon_{M_i}$ and $\varepsilon_{Y_i}$ are normally distributed errors with mean zero. Following the BCMF framework, we model $m(\cdot)$ and $y(\cdot)$ using BART, with separate models for the mediator and outcome regressions and weakly informative priors on the tree structures and leaf parameters. Following \citet{Linero2025-it}, we make two adjustments to standard BART priors that are critical in the mediation setting.

\subsection{Regularization-induced confounding}

Regularization-induced confounding (RIC) is a well-known challenge of using BART-based models for causal inference. In predictive settings, shrinkage can reduce variance and improve out-of-sample performance. In causal settings, however, the same regularization can bias causal contrasts when a treatment or mediator is associated with covariates that are also prognostic for the outcome \citep{Hahn2017-dv}. In mediation analysis, this concern arises most directly in the outcome regression, where the effect of the mediator must be separated from baseline covariates that predict both the mediator and the outcome. \citet{Linero2025-it} attribute RIC to what they call prior dogmatism, where default BART priors favor relatively simple response surfaces, which can implicitly favor models in which the covariates play a limited confounding role in the relationship between mediator and outcome.

In mediation analysis the concern is twofold. First, the treatment $A$ is associated with covariates that also predict the mediator and the outcome. Second, and most directly, the mediator $M$ is not randomized, so its effect on $Y$ must be separated from baseline covariates that predict both $M$ and $Y$. We address these in turn by including an estimated propensity score and a set of clever covariates, following \citet{Linero2025-it}.

\subsubsection{Propensity scores}

To address the first concern, we include an estimate of the propensity score
\[\pi(X_i) = P(A_i = 1 \mid X_i)\]
estimated from the treatment and covariates alone, as a predictor in both the mediator model and the outcome model \citep{Hahn2017-dv}. Supplying $\hat\pi(X_i)$ as a one-dimensional summary of the covariates helps ensure the treatment's effects on the mediator and on the outcome are not influenced by confounding that the regularization would otherwise leave in.

\subsubsection{Clever covariates}

To mitigate the second form of regularization bias, we follow \citet{Linero2025-it} by including what are known as clever covariates in the outcome regression. Originally introduced as a tool in targeted maximum likelihood estimation, clever covariates in BCMFs are predictions of the mediator mean at each treatment level \citep{van-der-Laan2006-ru, Linero2025-it}.

To use clever covariates in practice, we first fit the mediator model to obtain posterior predictions of the conditional means
\[m_a(X_i) = E\{M_i \mid A_i = a,  X_i\}, \qquad a \in \{0,1\},\]
and then include $\left(\hat m_0(X_i), \hat m_1(X_i)\right)$ as additional predictors in the BART model for $Y_i(a,m)$. These clever covariates summarize how $X$ shifts the mediator under each treatment arm, helping to separate the prognostic structure from the treatment structure in the outcome regression. Empirically, this leads to reduced bias and more stable posterior inference for mediation effects in highly nonlinear, confounded settings. 

\subsubsection{Stratification on treatment}  

Rather than treating $A_i$ as one predictor among many in a single BART ensemble, we fit separate BART models for each arm, stratifying on $A$ to account for treatment interactions with $X$ when modeling $M$ and treatment interactions with $X$ and $M$ when modeling $Y$. Stratification captures treatment interactions with $X$ when modeling $M$, and treatment interactions with $X$ and $M$ when modeling $Y$. This yields treatment-specific BART response surfaces $s_y(m, 1, x)$ and $s_y(m, 0, x)$ for the outcome and $s_m(1, x)$ and $s_m(0,x)$ for the mediator. This technique is well suited to estimating average mediation effects \citep{Linero2025-it}, though if heterogeneous effects are of primary interest, the varying-coefficient parameterization that regularizes the effects directly may be preferable \citep{Ting2025-xm}.

\subsubsection{Randomization}

When treatment is randomized, treatment exchangeability is satisfied by design, so adjustment for $X$ is not needed to remove confounding of the relationship between treatment and mediator or treatment and outcome. However, randomization of $A$ does not remove confounding of the relationship between the mediator and outcome, because the mediator is measured after treatment and not itself randomized. The outcome regression must therefore still adjust flexibly for baseline covariates that may jointly predict $M$ and $Y$. In observational studies, treatment exchangeability would additionally require that all confounders of the relationship between treatment and outcome, and treatment and mediator be measured and appropriately adjusted for. 

\subsection{G-formula for rule-level effects}

Let $\theta$ index the collection of unknown model parameters of the BART models for the mediator and outcome. In other words, $\theta$ denotes all decision trees in the BART models so that $\theta$ governs the conditional distributions of $M$
and $Y$. Then $f_\theta(Y_i = y \mid M_i = m, A_i = a, X_i = x)$ is the conditional density of the outcome and $f_\theta(M_i = m \mid A_i = a, X_i = x)$ that of the mediator. Under assumptions A1-A4, the marginal distribution of the cross-world counterfactuals $Y_i\{a, M_i(a')\} : a,a'\in\{0,1\}$ is identified by the g-formula representation
\begin{align*}
    &f_\theta(Y_i\{a,M_i(a')\}= y\mid X_i = x)\\
    &\qquad = \int f_\theta(Y_i = y \mid M_i = m, A_i = a, X_i = x) \\
    &\qquad\qquad \times f_\theta(M_i = m\mid A_i = a', X_i = x)dm
\end{align*}
for all $x\in\mathcal X$ and $a, a'\in\{0,1\}$ \citep{Imai2010-nc}.

To average over the covariates we place a prior on the covariate distribution $F_X$, following \citet{Linero2025-it}. Two choices are the empirical distribution, which fixes the weights at $\omega_i = 1/n$, and the Bayesian bootstrap \citep{Rubin1981-af}, which treats them as random with $\omega = (\omega_1, \ldots, \omega_n) \sim \mathrm{Dirichlet}(1, \ldots, 1)$ independently of the remaining parameters. In both cases $\omega_i \ge 0$ and $\sum_{i=1}^n \omega_i = 1$. Then, we can estimate the counterfactual expectation for a given treatment combination as
\begin{align*}
    E_\theta[Y_i\{a, M(a')\}] &= \sum_{i=1}^N \omega_i \int y f_\theta(Y_i = y \mid M_i = m, A_i = a, X_i)\\
    &\qquad\times f_\theta(M_i = m \mid A_i = a', X_i) dm\,dy.
\end{align*}

The estimand of interest, however, is not a static contrast at fixed treatment levels but a rule-level counterfactual mean, in which treatments are assigned by decision rules rather than held at constant values. Then, for a fixed rule pair $(d, r)$ with $d, r \in \mathcal{D}$, the rule-level counterfactual mean is
 \begin{align}
    E_\theta[Y^{d,r}(X)] &= E_\theta[Y_i\{d(X), M(r(X))\}]\notag\\
    &= \sum_{i=1}^N \omega_i \int s_y(Y_i = y \mid M_i = m, A_i = d(X_i), X_i)\notag\\ 
    &\quad \times f_\theta(M_i = m \mid A_i = r(X_i), X_i) dm
    \label{eq:gformula}
\end{align}

Repeating this procedure for \(Y^{d,d}(X)\), \(Y^{r,r}(X)\), \(Y^{d,r}(X)\), and \(Y^{r,d}(X)\) yields posterior draws of \(V(d)\), \(V(r)\), and the corresponding direct and indirect rule-effect components. While \eqref{eq:gformula} doesn't have a closed form when $s_y$ and $s_m$ are estimated using BART, we can approximate it by simulating many pseudodatasets. These pseudodatasets include the modeled estimate of the mediator under both treatment levels. To explain this process, let $F_M^-(\,\cdot \mid A=a, X_i)$ denote the generalized inverse CDF of the mediator distribution at the current posterior draw of $\theta$. For each replicate $k = 1, \ldots, K$, we draw a single uniform random variable $U_{ik}\sim\mathrm{Uniform}(0,1)$ per individual and use it to define
\[M^*_{ik}(a) = F_M^-\bigl(U_{ik} \mid A = a, X_i\bigr), \qquad a \in \{0,1\}.\]  

Because the same $U_{ik}$ is used at both treatment levels, the draws $M^*_{ik}(0)$ and $M^*_{ik}(1)$ are comonotone and are highly correlated \citep{Puccetti2015-fc}. This aids estimation because cross-world differences such as $s_y\{M^*_{ik}(d), \cdot, X_i\} - s_y\{M^*_{ik}(r), \cdot, X_i\}$ have the same noise, thereby decreasing the Monte Carlo variance. Then, with our estimates of $M^*_{ik}(a')$, we obtain estimates of the direct and indirect rule components $\Delta_D(d,r)$ and $\Delta_I(d,r)$
\begin{align*}
\Delta_D(d,r) &= \frac{1}{K} \sum_{k=1}^K \sum_{i=1}^N \omega_i \Big(s_y\big[M^*_{ik}\{r(X_i)\}, d(X_i), X_i\big] - s_y\big[M^*_{ik}\{r(X_i)\}, r(X_i), X_i\big]\Big)\\
\Delta_I(d,r) &= \frac{1}{K} \sum_{k=1}^K \sum_{i=1}^N \omega_i \Big(s_y\big[M^*_{ik}\{d(X_i)\}, r(X_i), X_i\big]- s_y\big[M^*_{ik}\{r(X_i)\}, r(X_i), X_i\big]\Big).
\end{align*}

\begin{algorithm}[t]
\caption{Monte Carlo g-formula for rule-level mediation}
\label{alg:gformula_rule}
\begin{algorithmic}[1]
\Require $\theta$, $K$, and $X_1,\ldots, X_N$
\For{$k=1,\ldots,K$}
  \For{$i=1,\ldots,N$}
    \State Sample $U_{ik}\sim \mathrm{Uniform}(0,1)$
    \State Set $M^*_{ik}\{d(X_i)\}=F_M^-\!\left(U_{ik}\mid A=d(X_i),X_i\right)$
    \State Set $M^*_{ik}\{r(X_i)\}=F_M^-\!\left(U_{ik}\mid A=r(X_i),X_i\right)$
  \EndFor
\EndFor
\State Compute the four rule-level counterfactual mean estimates
\begin{align*}
\widehat E[Y^{d,d}(X)] &= \frac{1}{K}\sum_{k=1}^K \sum_{i=1}^N \omega_i s_y\left(M^*_{ik}\{d(X_i)\},\, d(X_i),\, X_i\right),\\
\widehat E[Y^{d,r}(X)] &= \frac{1}{K}\sum_{k=1}^K \sum_{i=1}^N \omega_i s_y\left(M^*_{ik}\{r(X_i)\},\, d(X_i),\, X_i\right),\\
\widehat E[Y^{r,d}(X)] &= \frac{1}{K}\sum_{k=1}^K \sum_{i=1}^N \omega_i s_y\left(M^*_{ik}\{d(X_i)\},\, r(X_i),\, X_i\right),\\
\widehat E[Y^{r,r}(X)] &= \frac{1}{K}\sum_{k=1}^K \sum_{i=1}^N \omega_i s_y\left(M^*_{ik}\{r(X_i)\},\, r(X_i),\, X_i\right).
\end{align*}
\State Return $\hat\Delta_{D}(d,r) = \widehat E[Y^{d,r}(X)] - \widehat E[Y^{r,r}(X)]$ and $\hat\Delta_{I}(d,r) = \widehat E[Y^{d,d}(X)] - \widehat E[Y^{d,r}(X)]$.
\end{algorithmic}
\end{algorithm}

Repeating Algorithm \ref{alg:gformula_rule} over posterior draws and bootstrap weights $\omega$ yields posterior samples for $V(d)$, $V(r)$, $\Delta_D(d,r)$, and $\Delta_I(d,r)$. In practice, the uniform construction in Algorithm \ref{alg:gformula_rule} makes the Monte Carlo approximation highly stable, so that modest values of $K$ are typically sufficient.

\section{Simulation Study}
\label{s:simulation}

\begin{table*}[t]
\centering
\footnotesize
\setlength{\tabcolsep}{3pt}
\renewcommand{\arraystretch}{1.1}
\begin{tabularx}{\textwidth}{r *{4}{>{\centering\arraybackslash}X} *{4}{>{\centering\arraybackslash}X} *{4}{>{\centering\arraybackslash}X}}
\toprule
  & \multicolumn{4}{c}{\textit{Total}}
  & \multicolumn{4}{c}{$\widehat{\Delta}_D$}
  & \multicolumn{4}{c}{$\widehat{\Delta}_I$} \\
\cmidrule(lr){2-5}\cmidrule(lr){6-9}\cmidrule(lr){10-13}
$n$
  & Bias & RMSE & Cov & Width
  & Bias & RMSE & Cov & Width
  & Bias & RMSE & Cov & Width \\
\midrule
\multicolumn{13}{@{}l}{\textit{Scenario 1 (100\% Direct and 0\% Indirect)}} \\
  \phantom{0}400 & 0.01 & 0.13 & 0.97 & 0.51 & 0.00 & 0.06 & 0.97 & 0.30 & 0.01 & 0.09 & 0.99 & 0.42 \\
  \phantom{0}800 & 0.00 & 0.08 & 0.97 & 0.35 & 0.00 & 0.04 & 0.97 & 0.19 & 0.01 & 0.07 & 0.96 & 0.31 \\
  1200           & 0.00 & 0.07 & 0.95 & 0.28 & 0.00 & 0.04 & 0.96 & 0.15 & 0.00 & 0.06 & 0.96 & 0.25 \\
\addlinespace[4pt]
\multicolumn{13}{@{}l}{\textit{Scenario 2 (75\% Direct and 25\% Indirect)}} \\
  \phantom{0}400 & 0.01 & 0.13 & 0.95 & 0.51 & 0.02 & 0.06 & 0.98 & 0.28 & 0.00 & 0.10 & 0.99 & 0.43 \\
  \phantom{0}800 & 0.00 & 0.08 & 0.97 & 0.35 & 0.01 & 0.04 & 0.98 & 0.18 & 0.00 & 0.07 & 0.97 & 0.31 \\
  1200           & 0.00 & 0.07 & 0.96 & 0.28 & 0.01 & 0.03 & 0.97 & 0.14 & -0.01 & 0.06 & 0.97 & 0.25 \\
\addlinespace[4pt]
\multicolumn{13}{@{}l}{\textit{Scenario 3 (50\% Direct and 50\% Indirect)}} \\
  \phantom{0}400 & 0.01 & 0.13 & 0.96 & 0.51 & 0.02 & 0.06 & 0.98 & 0.27 & -0.01 & 0.10 & 0.99 & 0.44 \\
  \phantom{0}800 & 0.00 & 0.08 & 0.97 & 0.35 & 0.01 & 0.04 & 0.98 & 0.17 & -0.01 & 0.08 & 0.97 & 0.32 \\
  1200           & 0.00 & 0.07 & 0.97 & 0.28 & 0.01 & 0.03 & 0.98 & 0.13 & -0.01 & 0.06 & 0.95 & 0.26 \\
\addlinespace[4pt]
\multicolumn{13}{@{}l}{\textit{Scenario 4 (25\% Direct and 75\% Indirect)}} \\
  \phantom{0}400 & 0.01 & 0.13 & 0.96 & 0.51 & 0.03 & 0.06 & 0.95 & 0.27 & -0.02 & 0.11 & 0.97 & 0.45 \\
  \phantom{0}800 & 0.00 & 0.09 & 0.96 & 0.35 & 0.01 & 0.04 & 0.98 & 0.17 & -0.01 & 0.08 & 0.98 & 0.32 \\
  1200           & 0.00 & 0.07 & 0.96 & 0.28 & 0.01 & 0.03 & 0.96 & 0.13 & -0.01 & 0.06 & 0.96 & 0.26 \\
\addlinespace[4pt]
\multicolumn{13}{@{}l}{\textit{Scenario 5 (0\% Direct and 100\% Indirect)}} \\
  \phantom{0}400 & 0.01 & 0.13 & 0.95 & 0.51 & 0.04 & 0.07 & 0.95 & 0.28 & -0.03 & 0.11 & 0.98 & 0.47 \\
  \phantom{0}800 & 0.01 & 0.09 & 0.97 & 0.36 & 0.02 & 0.04 & 0.99 & 0.18 & -0.01 & 0.08 & 0.97 & 0.34 \\
  1200           & 0.00 & 0.07 & 0.96 & 0.29 & 0.01 & 0.03 & 0.96 & 0.13 & -0.01 & 0.07 & 0.95 & 0.27 \\
\bottomrule
\end{tabularx}
\caption{%
  Simulation performance of the BCMF estimator across scenarios and sample sizes ($S = 200$ replicates per cell). For each estimand, we report Monte Carlo bias (Bias), root mean squared error (RMSE), empirical 95\% credible interval coverage (Cov), and mean 95\% credible interval width (Width). The total effect oracle is fixed at $\textit{Total}^{*} = 0.500$ across all scenarios; the direct/indirect oracle split varies as shown.
}
\label{tab:sim_performance}
\end{table*}

We consider five simulation scenarios under varying shares of direct and indirect influence on the outcome. These scenarios are: all direct (100\% direct and 0\% indirect), primarily direct (75/25), balanced (50/50), primarily indirect (25/75), and all indirect (0/100). The simulation will be assessed by the method's ability to recover both the raw direct and indirect rule effects as well as their proportions relative to the estimate of the total effect. 

\subsection{Data Generation}

We generated data with $p=6$ baseline covariates, a binary treatment assignment, a continuous mediator, and a continuous outcome. Covariates $X = (X_1,\ldots, X_6)^\top$ were drawn from a multivariate normal distribution $MVN(\mu,\Sigma)$ with $\mu = 0$ and an AR(1) covariance structure such that $\Sigma_{jk}= \rho^{\abs{j-k}}$ where $\rho = 0.1$, reflecting mild pairwise correlation among covariates. To reflect the environment of a randomized controlled trial and satisfy exchangeability assumptions, treatment was assigned independently of $X$. The three levels of sample sizes are $N \in \{400, 800, 1200\}$. 

The mediator was generated from a treatment model $M\mid A, X \sim Normal(\mu_{M},\sigma_M^2)$ where $\mu_M = \beta_0 + X^\top \beta + (\beta_A + X^\top\beta_{AX})A$. Similarly, the outcome was generated from $Y\mid A, M, X \sim Normal(\mu_{Y}, \sigma_Y^2)$ where $\mu_Y = \theta_0 + X^\top\theta + (\theta_A + X^\top\theta_{AX})A + \theta_M M + \gamma(X_1^2 - 1)$. The quadratic term $\gamma(X_1^2 - 1)$ introduces nonlinearity meant to reflect realistic clinical measures such as sleep duration.

We compared an active treatment rule $d$ to a reference rule $r$. The rule $d$ was a threshold rule that assigns treatment when a weighted linear index of four covariates exceeds a threshold
\[d(X) = I(X_1 + 0.8 X_2 + 0.2 X_3 + 0.1 X_4 \geq 0.25).\]
Weights are designed to be beneficial but imperfect, reflecting the type of rule one might realistically recover using precision medicine techniques such as Q-learning or outcome weighted learning. Therefore, the chosen rule approximates, without itself being, the optimal oracle rule under the given data generating process. The reference rule treats no one, assigning treatment value $A=0$ to all subjects. 

\begin{figure*}
  \centering
  \includegraphics[width=\linewidth]{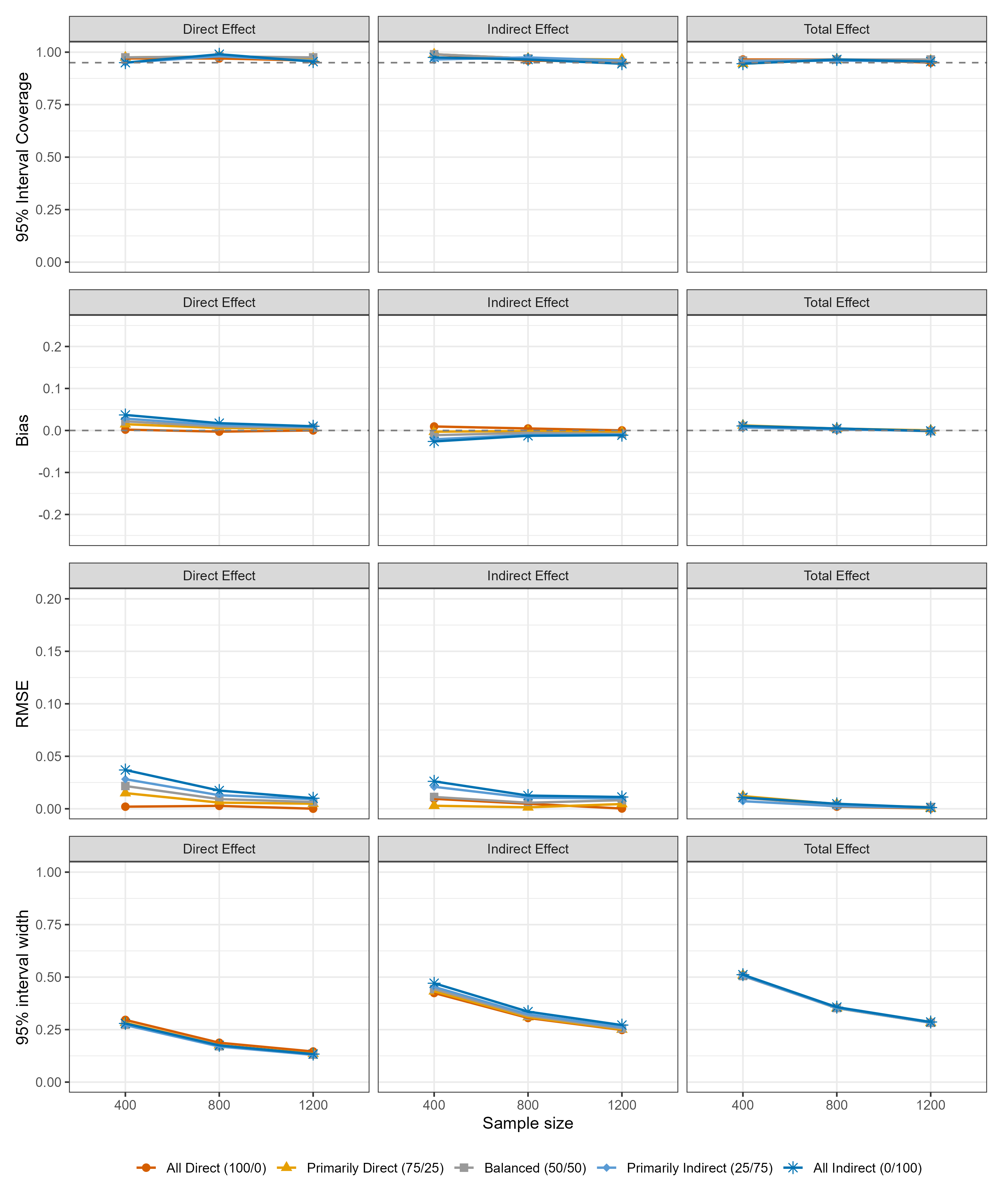} 
  \caption{95\% confidence interval coverage, bias, RMSE, and interval width for all mediation scenarios and sample sizes}
  \label{fig:metrics}
\end{figure*}

\subsection{Simulation Performance}

We ran $S=200$ independent replicates for each combination of the five scenarios and three sample sizes. Within each replicate, a dataset of size $n$ was generated, and the value for the reference rule and rule of interest were estimated. Furthermore, the value contrasts of the two rules were derived, decomposed into indirect and direct effects, and compared to the scenario-specific oracle value. It should be noted that decompositions are sensitive to the choice of the reference rule and the rule of interest by construction. Because the decomposition reflects how a particular rule reallocates treatment, two rules over the same population can produce different splits of the direct and indirect effect. In the extreme, a rule that never treats anyone whose treatment would shift the mediator yields no indirect effect, even when a mediated pathway is present in the population.

Simulation performance was assessed in terms of bias, root mean squared error (RMSE), empirical 95\% credible interval coverage of the oracle, and credible interval width. These results are reported in Table \ref{tab:sim_performance} and visualized in Figure \ref{fig:metrics}. 

Across all scenarios and sample sizes, the simulations achieved nominal coverage for all estimands. At smaller sample sizes, the direct rule effect tended to demonstrate a slight bias upward and the indirect rule effect a slight bias downward in more heavily mediated scenarios. The magnitude of this bias decreased as the sample size increased. Similarly, the RMSE was larger for both the direct and indirect rule effects in more heavily mediated scenarios, but decreased to values similar to those of less heavily mediated scenarios as the sample size increased. Finally, the 95\% credible interval width was similar across all scenarios for all estimands, with interval width decreasing with sample size as expected. Notably, the interval width is slightly wider for the indirect rule effect than for the direct rule effect across all scenarios and sample sizes.

In addition to recovering the raw estimates of the direct and indirect rule effects, one may also be interested in recovering the proportions of direct and indirect effects relative to the estimated total effect. Figure \ref{fig:split} shows the estimated proportion of direct and indirect shares relative to the oracle value as a dashed line for the $n=1200$ case. As we can see, in finite samples, there is a small attenuation of the indirect effect, which in this case is about a 2\% smaller than the true value.  

\begin{figure*}
  \centering
  \includegraphics[width=\linewidth]{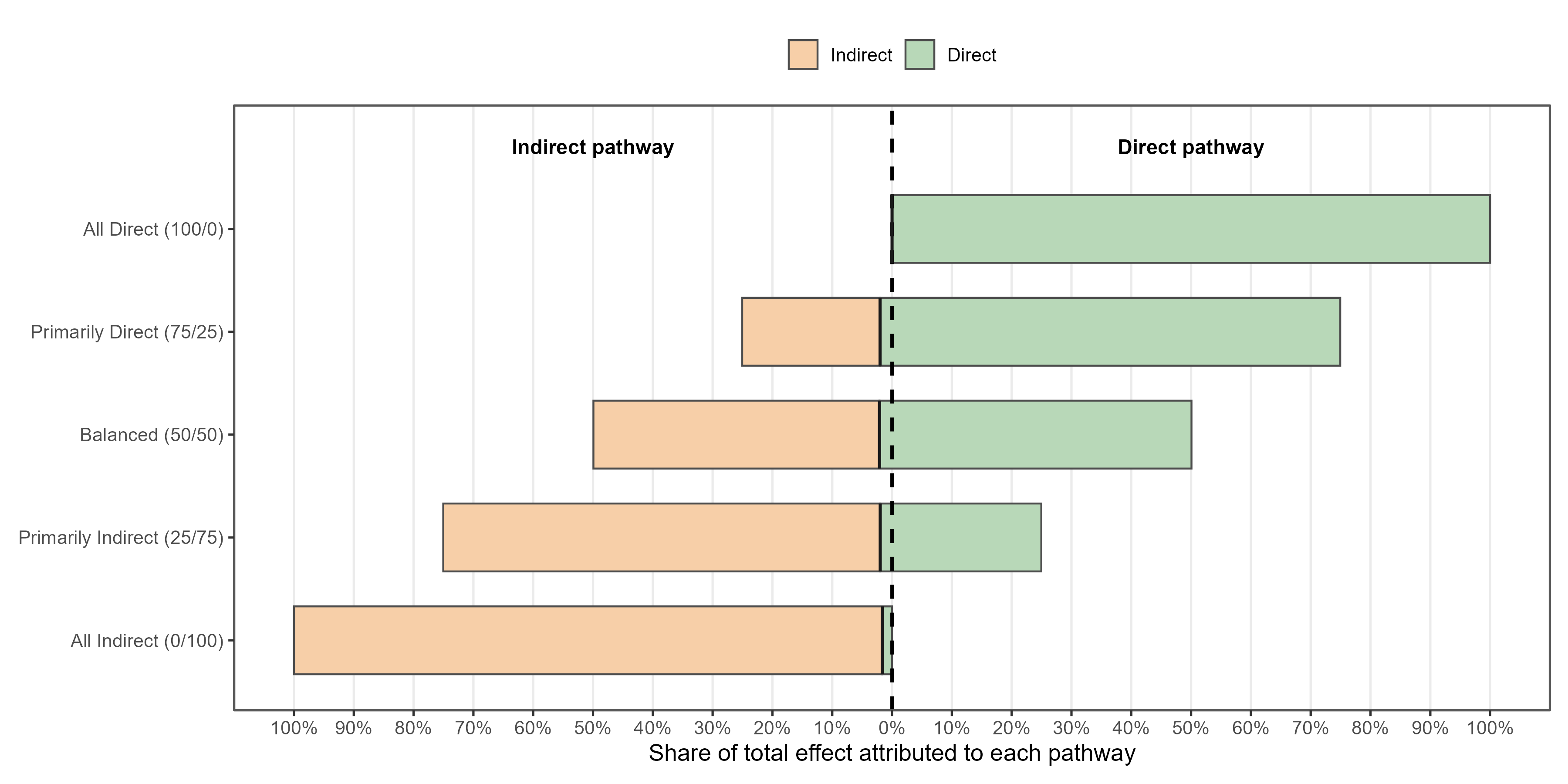} 
  \caption{Proportion of estimated direct and indirect effect for the $n = 1200$ case. Solid lines represent the estimated share, while the dotted line denotes the oracle share.}
  \label{fig:split}
\end{figure*}

\section{Real Data Analysis - TRIUMPH Trial}
\label{s:application}

We applied our mediation framework to the data from the Treating Resistant Hypertension Using Lifestyle Modification (TRIUMPH) trial, a randomized controlled trial comparing an intensive lifestyle program combining diet, behavioral weight management, and supervised aerobic exercise (C-LIFE) against a standardized education and physician advice (SEPA) control \citep{Blumenthal2015-pk}. Patients ($N = 140$) were randomized 2:1 to C-LIFE versus SEPA over a four-month intervention period. Although the parent trial was designed to assess blood pressure outcomes, ancillary measurements collected cognitive executive function as an exploratory outcome of interest. Supporting analyses further examined the relationship between cognitive function and candidate vascular, metabolic, and behavioral mechanisms \citep{Smith2022-et}, making TRIUMPH a useful setting for illustrating mediation analysis for individualized treatment rules. 

The motivating scientific question is not only whether C-LIFE improves cognitive outcomes, but also whether an individualized rule that recommends C-LIFE to high risk participants produces its cognitive benefit through specific causal mechanisms. Prior TRIUMPH analyses suggest that cognitive response to lifestyle modification may be modified with respect to baseline inflammatory and vascular risk profiles \citep{Avorgbedor2023-wj, Smith2022-et}. In addition, TRIUMPH collected several post-treatment measures plausibly related to cognitive benefit, including cerebrovascular oxygenation, aerobic fitness, and sleep quality \citep{Smith2022-hx, Smith2023-xk}. We use the proposed decomposition to ask whether the value improvement of a candidate treatment rule over a reference rule that assigns SEPA to everyone is explained by any one of these candidate mediators.

Let $A \in \{0, 1\}$ denote randomized assignment to SEPA (0) or C-LIFE (1), $X$ denote baseline covariates, $M$ denote a candidate mediator measured at four months, and $Y$ denote cognitive executive function measured at four months (with higher values indicating improvement). The estimands of interest are
\begin{align*}
  \Delta_D(d, r) &= E[Y^{d,r}(X)] - E[Y^{r,r}(X)]\quad\textrm{and}\\
  \Delta_I(d, r) &= E[Y^{d,d}(X)] - E[Y^{d,r}(X)],
\end{align*}
where $r$ is the reference rule ($r(X) \equiv 0$) and $d$ is a clinically motivated treatment rule. We restricted to participants with complete data for the outcome, candidate mediator, and all 12 baseline covariates. Depending on the choice of mediator, this yielded between n=113 and n=120 participants. Baseline covariates included age, sex, education, waist-to-hip ratio, baseline ambulatory SBP, baseline ambulatory DBP, baseline peak VO2, baseline CRP, metabolic risk z-score (a composite of BP, lipids, and adiposity), baseline executive function, baseline processing speed, and baseline memory.

The value of a rule is the expected executive function score that would be observed if treatment were assigned according to that rule. The total value improvement can be represented as $E[Y^{d,d}(X)] - E[Y^{r,r}(X)]$, which compares the expected outcome under the candidate ITR with the expected outcome under the reference rule. Because each mediator is analyzed separately, the direct component should be interpreted as all pathways not operating through that mediator, including pathways through other measured or unmeasured mechanisms.

The ITR of interest $d$ assigns C-LIFE to patients with high vascular and inflammatory burden. The rule uses a composite score that thresholds a linear combination of high baseline CRP (50\%), low aerobic fitness (30\%), and high ambulatory SBP (20\%), each standardized with zero mean and unit variance within the analysis sample. Patients scoring above 0 (corresponding to an above-average composite risk) are recommended C-LIFE. The remainder receives SEPA. The choice of variable and weighting reflects the TRIUMPH literature's ranking of effect modifiers \citep{Avorgbedor2023-wj, Smith2022-et, Blumenthal2015-pk}, rendering an empirically derived rule. 

This rule should be interpreted as an empirically motivated rule informed by literature rather than as an estimated optimal rule. The zero threshold corresponds to above average composite risk in the analysis sample. Consequently, the rule assigned approximately 43-46\% of participants to C-LIFE across the complete case samples for each mediator. This allocation proportion is a consequence of the chosen threshold and should not be interpreted as a general requirement that an ITR allocate treatment to half of the target population.

\subsection{Mediation Scenarios}

As mentioned earlier, we evaluated three candidate mediators: change in oxygenated hemoglobin during cognitive challenges assessed using functional near-infrared spectroscopy (fNIRS), change in aerobic fitness assessed using VO2 max, and change in sleep quality. These mediators were chosen to represent three distinct neurovascular, metabolic/cardio-respiratory, and behavioral mechanisms through which the treatment may affect the outcome \citep{Smith2022-hx, Avorgbedor2023-wj, Smith2023-xk}. Each mediator was analyzed separately. Consequently, the direct component in a given analysis should be interpreted as the component not operating through that specific mediator. It may still include pathways through other measured or unmeasured mediators.

Appropriate mediators must occur post-treatment and pre-outcome in order to lie on the causal pathway between $A$ and $Y$. Within the TRIUMPH trial, all three mediators were collected as a change from baseline (post-treatment) at a four-month assessment. Cognitive benefit was also observed at this four-month follow-up, raising potential concerns about simultaneity and reverse causality. However, while oxygenated hemoglobin and VO2 max are measured at the same visit as cognitive benefit, they reflect a biologically antecedent state. Additionally, sleep quality is a self-reported measure recording the quality of sleep over the preceding month. Therefore, sleep is a retrospective measure that resolves the simultaneity problem. 

\begin{table}[t]
\centering
\setlength{\tabcolsep}{4pt}
\renewcommand{\arraystretch}{1.15}
\begin{tabular*}{\columnwidth}{@{}l@{\extracolsep{\fill}}ccc@{}}
\toprule
  & \shortstack[c]{\textit{fNIRS oxygenation}}
  & \shortstack[c]{\textit{Aerobic fitness}}
  & \shortstack[c]{\textit{Sleep quality}} \\[3pt]
\midrule
\multicolumn{4}{@{}l}{\textit{Sample}} \\
\qquad $n$                                    & 116       & 120       & 113       \\
\qquad C-LIFE\,/\,SEPA                       & 82\,/\,34 & 85\,/\,35 & 81\,/\,32 \\
\qquad Percent assigned to rule (\%)              & 45.7      & 42.5      & 46.0      \\[3pt]
\multicolumn{4}{@{}l}{\textit{ITR outcomes}} \\
\qquad Value under rule, $\hat{V}(d)$          & $65.8^{*}$ & $65.8^{*}$ & $64.8^{*}$ \\
\qquad Value under reference, $\hat{V}(r)$     & $62.1^{*}$ & $62.1^{*}$ & $60.3^{*}$ \\[3pt]
\multicolumn{4}{@{}l}{\textit{Effect decomposition}} \\
\qquad Total benefit, $\hat{V}(d)-\hat{V}(r)$                          & $3.66^{*}$ & $3.76^{*}$ & $4.52^{*}$ \\
\qquad Direct effect, $\widehat{\Delta}_D(d,r)$    & $3.43$     & $3.48^{*}$ & $4.41^{*}$ \\
\qquad Indirect effect, $\widehat{\Delta}_I(d,r)$  & $0.23$     & $0.28$     & $0.11$     \\[3pt]
\multicolumn{4}{@{}l}{\textit{Mediated share (\%)}} \\
\qquad Direct share             & $96.5^{*}$ & $93.1^{*}$ & $97.7^{*}$ \\
\qquad Indirect share           & $3.5$      & $7.0$      & $2.3$      \\
\bottomrule \\[-8pt]
\end{tabular*}
\caption{%
  Mediation decomposition of the ITR benefit on executive function for three candidate mediating mechanisms: cerebral oxygenation, aerobic fitness, and sleep quality. All estimates are posterior means. Asterisks denote that the posterior 95\% credible interval excludes zero.
}
\label{tab:real_data}
\end{table}

\subsection{Results}

Table \ref{tab:real_data} reports posterior mean estimates for each analysis. Across the three complete-case samples, the rule assigned approximately 43-46\% of participants to C-LIFE. The estimated value under the rule exceeded the value under the reference rule for all three mediators, with estimated total benefits ranging from 3.66 to 4.52 points on the executive function scale.

For the given rule of interest, the benefit of the rule of interest over the treat-none rule is primarily realized via the direct pathway, with shares of 96.5\%, 93.1\%, and 97.7\% for oxygenated hemoglobin, aerobic fitness, and sleep quality, respectively. In contrast, the estimated indirect components were small: 3.5\% for oxygenated hemoglobin, 7.0\% for aerobic fitness, and 2.3\% for sleep quality. Therefore, the decomposition suggested that most of the estimated rule benefit was not mediated by any single candidate mediator considered here, as the estimated indirect components for cerebral oxygenation, aerobic fitness, and sleep quality were small. The corresponding credible intervals included zero, indicating that we cannot conclude that there was a significant contribution along the mediated pathway for any of the three candidate mediators. 

However, this pattern should not be interpreted as evidence that lifestyle modification has no mechanistic effects, as the analysis is relative to the choice of ITRs $d$ and $r$. Furthermore, because each mediator was considered separately, the direct component includes all pathways not operating through the specific mediator under analysis, including pathways through other physiological or behavioral mechanisms. It is also important to note that while the rule $d$ realized an increase in value relative to $r$ across all three scenarios, this margin was small relative to the value delivered by $r$. The total value improvement of $d$ relative to $r$ is the quantity being decomposed. When this total improvement is small, there is limited benefit to explain mechanistically. Consequently, even a mediator that accounts for a nontrivial share of the total improvement may correspond to a small absolute indirect effect. The small magnitude of the total effect could partly be due to the fact that $r$ is not a true ``treat-none" rule as it assigned SEPA to everyone. Assuming SEPA has a positive value, the increase in value relative to a rule that truly assigns no intervention would be higher. In other settings, we would hope that the chosen ITR would realize a greater benefit over the rule that treats no one.    

Overall, the TRIUMPH application illustrates how mediation for ITRs can move beyond estimating the value of an individualized rule to ask how that value is produced. In this example, the estimated cognitive benefit of the rule appeared to be driven primarily by pathways other than the three individual mediators considered. Future applications with larger samples and stronger rule contrasts may be better positioned to identify mechanistic pathways underlying ITR value.

\section{Discussion}
\label{s:discussion}

In this paper, we propose a mediation framework for individualized treatment rules that decomposes a contrast of the value function into direct and indirect components. Unlike classical mediation analysis, which focuses on decomposing a fixed treatment effect, our approach targets mechanistic explanations of why a candidate rule outperforms a reference rule. We formalize rule-level potential outcomes $\{Y^{d,d}(X),Y^{r,r}(X),Y^{d,r}(X),Y^{r,d}(X)\}$, showing that the corresponding decomposition is identified under standard identifiability assumptions, and propose a Bayesian estimation strategy that adapts Bayesian causal mediation forests for ITR evaluation.

The key methodological innovation is the replacement of fixed treatment levels by rule-induced actions, coupled with a Monte Carlo implementation that produces posterior samples of $V(d)$, $V(r)$, $\Delta_D(d,r)$, and $\Delta_I(d,r)$. The use of stratification on $A$ and clever covariates follows prior BCMF recommendations to stabilize inference in the presence of complex confounding and nonlinear interactions. From an applied perspective, the proposed estimation strategy yields an interpretable decomposition of the value improvement delivered by an ITR relative to a clinically meaningful baseline.

Several limitations and practical considerations remain. First, the interpretation of the decomposition hinges on assumptions A2-A3 of no unmeasured confounding, which cannot be verified in observational settings. Because of this, in practice, sensitivity analyses are essential. Second, positivity is critical for rule evaluation. If a rule assigns actions in the covariate space with little or no support in the observed data, inference will rely on extrapolation, and credible intervals may understate uncertainty. Third, when $d$ is learned from the same dataset used to evaluate mediation, additional variability and potential optimism may arise. However, this can be alleviated using cross-validation or a sample splitting paradigm.

Future work may extend the framework to dynamic treatment regimes, multiple mediators, and heterogeneous mediation. This approach advances precision medicine by complementing rule optimization with pathway-based explainability that can inform underlying mechanistic understanding.







\bibliographystyle{oup-abbrvnat}
\bibliography{paperpile}




\begin{appendices}

\section{Identification of Rule Mediation Estimands}\label{sec11}

\begin{lemma}[Identification of rule mediation estimands]
\label{lem:identification}
Suppose assumptions A1-A4 hold. Then for any pair of rules $d,r\in\mathcal D$, the cross-world rule-level mean is identified by the mediation g-formula,
\begin{align*}
  E\{Y^{d,r}(X)\} &= \int_{\mathcal X}\int_{\mathcal M} E(Y\mid A=d(x),M=m,X=x)\\
  &\qquad\times f_{M\mid A,X}\{m\mid A=r(x),X=x\}\,dm\,dF_X(x).
\end{align*}
In particular, taking $(d,r)$ equal to $(d,d)$, $(r,r)$, $(d,r)$, and $(r,d)$ identifies $V(d)$, $V(r)$, and the two cross-world means $E\{Y^{d,r}(X)\}$ and $E\{Y^{r,d}(X)\}$. Consequently the decomposition $V(d)-V(r)=\Delta_D(d,r)+\Delta_I(d,r)$ and each of its components are identified from the observed data distribution.
\end{lemma}

\begin{proof}
Fix $d,r\in\mathcal D$ and $x\in\mathcal X$. Positivity (A4) guarantees that the regression $E(Y\mid A=d(x),M=m,X=x)$ and the conditional density $f_{M\mid A,X}\{m\mid A=r(x),X=x\}$ are well defined on the relevant support. Write 
\begin{equation*}
    \mu(x)\equiv E\{Y^{d,r}(X)\mid X=x\}=E[Y\{d(x),M(r(x))\}\mid X=x].
\end{equation*}

Iterating the expectation over the potential mediator $M(r(x))$,
\begin{align*}
    \mu(x)&=\int_{\mathcal M} E\big[Y\{d(x),m\}\mid M(r(x))=m,\,X=x\big] dF_{M(r(x))\mid X}(m\mid x).
\end{align*}
By cross-world exchangeability (A3), $Y\{d(x),m\}\perp M(r(x))\mid X=x$, so the inner conditional expectation does not depend on the event $\{M(r(x))=m\}$ and can be expressed as
\begin{equation*}
    E\left[Y\{d(x),m\}\mid M(r(x))=m,X=x\right]=E\left[Y\{d(x),m\}\mid X=x\right].
\end{equation*}
For the outcome term, treatment exchangeability (A2) gives $Y\{d(x),m\}\perp A\mid X=x$, and therefore 
\begin{equation*}
    E[Y\{d(x),m\}\mid X=x]=E[Y\{d(x),m\}\mid A=d(x),X=x].
\end{equation*} 
Furthermore, the joint independence of A2 along with A3 yields $Y\{d(x),m\}\perp M(d(x))\mid A=d(x),X=x$. Since $M=M(d(x))$ on the event $A=d(x)$ by consistency (A1),
\begin{equation*}
    E\big[Y\{d(x),m\}\mid A=d(x),X=x\big]=E\big(Y\mid A=d(x),M=m,X=x\big).
\end{equation*}
For the mediator term, A2 gives $M(r(x))\perp A\mid X=x$, and consistency (A1) gives $M=M(r(x))$ on $A=r(x)$, so
\[dF_{M(r(x))\mid X}(m\mid x)=f_{M\mid A,X}\{m\mid A=r(x),X=x\}\,dm.\]
Substituting the last three equations into $\mu(x)$ and integrating over the marginal distribution of $X$ gives the stated g-formula. The remaining claims follow by noting the results hold without loss of generality to the choice of $(d,r)$ and noting that $\Delta_D(d,r)$ and $\Delta_I(d,r)$ are differences of the identified means.
\end{proof}

\end{appendices}

\end{document}